\documentclass[letterpaper, 10 pt, conference]{ieeeconf}  % Comment this line out if you need a4paperieeeconfIEEEtran

\pdfoutput=1
\IEEEoverridecommandlockouts                              % This command is only needed if 
\usepackage{cite}
\usepackage{amsmath,amssymb,amsfonts}
\usepackage{algpseudocode}
\usepackage{bm}
\usepackage{graphicx}
\usepackage{textcomp}
\usepackage{xcolor}
\usepackage{accents}
\let\labelindent\relax
\usepackage{enumitem}
\usepackage[caption=false,font=footnotesize]{subfig}
\newcommand{\ubar}[1]{\underaccent{\bar}{#1}}

\newcommand{\cvar}[0]{\textup{\textrm{CVaR}}}
\newcommand{\conv}[0]{\textup{\textrm{conv}}}
\newcommand{\floor}[1]{\lfloor #1 \rfloor}

\newtheorem{theorem}{Theorem}

\newtheorem{lemma}{Lemma}
\newtheorem{assumption}{Assumption}
\newtheorem{definition}{Definition}

\def\BibTeX{{\rm B\kern-.05em{\sc i\kern-.025em b}\kern-.08em
    T\kern-.1667em\lower.7ex\hbox{E}\kern-.125emX}}

\title{\LARGE \bf
Uncertainty-aware Flexibility Envelope Prediction in Buildings with Controller-agnostic Battery Models
}

\author{Paul Scharnhorst$^{1, 2, *}$, Baptiste Schubnel$^1$, Rafael E. Carrillo$^1$, Pierre-Jean Alet$^1$ and Colin N. Jones$^{2}$% <-this % stops a space
\thanks{This work received support from CSEM's Data Program and the Swiss National Science Foundation under the RISK project (Risk Aware Data Driven
Demand Response, grant number 200021 175627).}% <-this % stops a space
\thanks{$^{1}$Paul Scharnhorst, Baptiste Schubnel, Rafael E. Carrillo, and Pierre-Jean Alet are with  CSEM S.A., 2002 Neuchâtel, Switzerland, Email: \{{\tt paul.scharnhorst, baptiste.schubnel, rafael.carrillo, pierre-jean.alet}\}@csem.ch}%
\thanks{$^{2}$Paul Scharnhorst and Colin N. Jones are with LA, EPFL, 1015 Lausanne, Switzerland,  Email: {\tt colin.jones@epfl.ch}}%
\thanks{$^*$ Corresponding author}%
}

\begin{document}

\maketitle
\thispagestyle{empty}
\pagestyle{empty}

%%%%%%%%%%%%%%%%%%%%%%%%%%%%%%%%%%%%%%%%%%%%%%%%%%%%%%%%%%%%%%%%%%%%%%%%%%%%%%%%
\begin{abstract}
Buildings are a promising source of flexibility for the application of demand response. In this work, we introduce a novel battery model formulation to capture the state evolution of a single building. Being fully data-driven, the battery model identification requires one dataset from a period of nominal controller operation, and one from a period with flexibility requests, without making any assumptions on the underlying controller structure. We consider parameter uncertainty in the model formulation and show how to use risk measures to encode risk preferences of the user in robust uncertainty sets. Finally, we demonstrate the uncertainty-aware prediction of flexibility envelopes for a building simulation model from the Python library Energym.
\end{abstract}

%%%%%%%%%%%%%%%%%%%%%%%%%%%%%%%%%%%%%%%%%%%%%%%%%%%%%%%%%%%%%%%%%%%%%%%%%%%%%%%%
\section{Introduction}
Electrification, along with the increase of renewable energy production, serves as a means to mitigate climate change \cite{WEO_IEA}. However, increasing energy demand, together with the intermittent behavior of renewable energy sources like solar or wind power, creates new challenges for grid operators in maintaining a balanced grid. A key tool to achieve this balancing, is Demand Response (DR) \cite{zhou2021electrification}, a way for prosumers to adapt their consumption to available energy generation. 

Various studies have been undertaken on the topic of DR, from reviews on DR definitions and metrics \cite{REYNDERS2018372}, to socio-economic surveys on the acceptance of DR schemes \cite{Acceptance_DR}, to research on the probability and extent of consumer reaction to price signals \cite{GANESAN2022111663}. In this work, we will focus on direct DR, a setting where participating prosumers are remunerated for following explicit consumption signals from the grid operator, while not providing flexibility in a prespecified range leads to penalties. This is different from indirect DR, where participants are incentivized to adapt their consumption behavior through varying price signals.

Buildings have been identified as promising assets to provide flexibility (e.g. \cite{vrettos_robust_2014}). Equipped with e.g. heat pumps, photovoltaic systems, and electric vehicles, buildings have the potential to significantly change their consumption patterns in the short term. An exact estimation of the flexibility is necessary to efficiently use the buildings' capabilities for DR.

We will focus on the single building flexibility estimation problem in this work. As a tool for this estimation, we will use the concept of flexibility envelopes. Flexibility envelopes, as in \cite{DHULST201579}, provide availability time predictions of power level changes from a given baseline, depending on the time of the day. Reference \cite{gasser_predictive_2021} uses flexibility envelopes without considering a baseline, by reporting availability times for discretized power levels. While being able to adapt the control objectives, the approach relies on modeling the building and its systems. A data-driven way of learning and approximating flexibility envelopes is presented in \cite{hekmat2021datadriven}. This approach is capable of reducing the computational burden, but a training set of precomputed flexibility envelopes is needed.

Due to the potentially difficult and complex modeling of buildings and their equipment, virtual battery models have been determined as an efficient tool to model the thermal capacity of a building. Reference \cite{muller_autonomous_2017} provides a set description of the feasible power consumption profiles, with the use of a data-driven battery model, for heating systems with on/off behavior. Input tracking for constrained linear discrete-time systems is considered in \cite{gorecki_guaranteeing_2015}, with the aim of certifying input trackability for a given reference set. In a DR application, a fixed-shape reference set is parameterized by a battery model and optimized over to offer cost-optimal flexibility to the grid operator. Further usage of battery models for aggregations of thermostatically controlled loads is shown in e.g. \cite{zhao_geometric_2016}.

Our contributions are threefold. Firstly, we introduce a data-driven battery model for representing the state of a building, without assuming a fixed controller structure and present a method for its parameter identification. Secondly, parameter uncertainty is considered for the resulting set of feasible trajectories. Risk measures are used to formulate robust uncertainty sets that take risk preferences of the user into account. Thirdly, we use the battery model to compute flexibility envelopes for a simulated building from the Energym library \cite{scharnhorst_energym} and compare the results for different risk levels.

\textbf{Notation:} We denote the indicator function of the set $\{0\}$ by $\chi$, so $\chi_q = \begin{cases}1, \text{ if } q = 0 \\ 0, \text{ if } q \neq 0 \end{cases}$. We use bold symbols for vectors and trajectories, with trajectories written as $\bm{r}_{0:k-1} = [r_0, \dots, r_{k-1}]^\top$, or simply $\bm{r}$ if the context is clear. The $N$-dimensional probability simplex is given by $\Delta^N=\{\bm q\in\mathbb{R}^N:  q_i\geq0, i=1,\dots,N, \sum_{i=1}^N q_i=1\}$. $\mathbf{0}$ and $\mathbf{1}$ denote a vector of appropriate size of zeros or ones, respectively. We assume empty sums to be $0$ and $\mathbb{R}^0=\{0\}$.

\section{Learning Battery Models}
\label{sec:lbm}
\subsection{Model Formulation}
\label{subsec:model}
We assume that the state of a building at time $t$, in terms of its thermal capacity, can be described by a scalar $s_t$, bounded by $s_{\min}$ and $s_{\max}$ (w.l.o.g. $s_{\min}=0$ and $s_{\max}=1$). The state is a measure of the energy stored in the system, and its bounds depend on the thermal bounds which are related to comfort or operational constraints in the building. This means that $s_t=0$ indicates that no energy can be extracted from the building without violating constraints, and $s_t=1$ indicates that no energy can be inserted. A possible definition of this abstract state is given in Section~\ref{sec:statedef}.%For now, the state $s(t)$ is abstract, but we will introduce a meaningful definition later on. 

%################################################
%################################################

We now proceed to the derivation of a battery like state equation for the state $s_t$. In full generality, the state evolution can be represented by the following difference equation:
\begin{equation}
    s_{t+1} - s_t = h\left(s_t, \bm e_t, p_t\right) + \omega_t,
    \label{eq:sde}
\end{equation}
with the external weather conditions denoted by $\bm e_t$, which can comprise multiple measurements and past data, the power injected $p_t$, and a noise term $\omega_t$. The noise term accounts for random disturbances in the system, e.g. occupants.

We assume that the controller operating the system leads to the state following a specific pattern, here called ``nominal state'' and denoted by $s_{\text{n}, t}$. The evolution of this nominal state can similarly be expressed as 
\begin{equation}
    \label{eq:sdenom}
    s_{\text{n}, t+1} - s_{\text{n}, t} = h(s_{\text{n}, t}, \bm e_t, p_{\text{n}, t}) + \omega_t,
\end{equation}
with the baseline power injected given by $p_{\text{n}, t}$.

Prediction of this baseline power is not the focus of this work. An overview of data-driven methods to predict energy consumption in buildings can be found in \cite{datadrivenenergy}. To introduce uncertainty quantification for consumption prediction, methods like Gaussian Process (GP) regression \cite{GPforML}, kernel methods with error quantification \cite{MADDALENA2021109896}, or variational autoencoders \cite{ljung_vae} can be used. We will therefore assume to have reasonably accurate baseline predictions where necessary.

Considering the difference of \eqref{eq:sde} and \eqref{eq:sdenom}, we get
\begin{equation}
\label{eq:stateevol}
    s_{t+1} -s_t =  s_{\text{n}, t+1} - s_{\text{n}, t} + h(s_t, \bm e_t, p_t) - h(s_{\text{n}, t}, \bm e_t, p_{\text{n}, t}).
\end{equation}

The goal of our work is to quantify the system behavior, and therefore the evolution of $s_t$, in cases where the nominal controller actions are augmented by specific requests.
\begin{definition}[Relative Consumption Request]
Given a baseline power $p_{\text{n}, t}\in\mathbb{R}$, we define a relative request as $r_t\in\mathbb{R}$ such that the desired overall consumption of the building at time $t$ is $ p_t=p_{\text{n}, t}+r_t$.
\end{definition}

We consider systems with controllers that drive the state back to its nominal value after receiving requests (as typically observed in thermal assets), therefore, we get two distinct phases in the system operation:
\begin{enumerate}
    \item The request phase where $p_t = p_{\text{n}, t}+r_t$.
    \item The recovery phase where $s_t$ is driven towards $s_{\text{n}, t}$, with the injected power denoted by $p_{\text{con}, t}$.
\end{enumerate}
Moreover, in the request phase, we can distinguish between receiving positive or negative relative consumption requests, due to equipment or controller characteristics.

We make the following assumption about the controller.
\begin{assumption}
\label{as:statebound}
For each $s_t\in [0,1]$, the controller is able to satisfy the comfort/operational constraints for all $t'>t$. When receiving flexibility requests, the controller follows them as closely as possible, without violating constraints. Furthermore, we assume to either receive state measurements from the controller, or measurements from which we can construct a state-like variable.
\end{assumption}

Note that Assumption \ref{as:statebound} does not impose a fixed controller, and therefore, is very general in its application. The assumption on fulfilling constraints is more an assumption on the equipment than the controller since any decent controller should be able to fulfill constraints with enough controllability. Lastly, not relying on a fixed state definition further increases generality, while still having the option to construct a state from standard measurements (see Section~\ref{sec:statedef}).

Through Assumption~\ref{as:statebound}, we have that the overall state dynamics behave like a switched system, distinguishing the cases where $p_t=p_{\text{n}, t} + r_t$ and $p_t=p_{\text{con}, t}$. Assuming a linear approximation of $h$, for simplicity, around the nominal operation point, we get that
\begin{equation}
\begin{split}
    h(s_t, \bm e_t, p_t) - & h(s_{\text{n}, t}, \bm e_t, p_{\text{n}, t}) \\
    &\approx \begin{cases}
        a^{+} r_t \text{ if } r_t>0 \\
        a^{-} r_t \text{ if } r_t<0 \\
        b_f (s_t-s_{\text{n},t}) \text{ if } r_t=0
    \end{cases}.
\label{linear1}
\end{split}
\end{equation}
Due to the stochasticity of $s_t$, notice that $a^{+}$, $a^{-}$ and $b_f$ are in general stochastic.

We make a few further assumptions on the nominal state evolution and the coefficients $a^{+}$, $a^{-}$, and $b_f$ that will ease the rest of the analysis.
\begin{assumption}
\label{as:discreteassum}
In \eqref{linear1}, we assume that
\begin{enumerate}[label=(\alph*)]
    \item The request-free nominal state evolution $s_{\text{n}, t}$ can be well approximated by a function  $f:\mathbb{R}^{m} \to \mathbb{R}$ of the current and recent past weather variables, denoted hereafter by
    $\bm{e}_t := [\bm e_{1,t}^\top,..., \bm e_{n,t}^\top]^\top \in \mathbb{R}^{m}, \bm e_{i,t}\in \mathbb{R}^\eta , i=1,\dots, n, \quad m=n\eta$,
    
    \item $b_f\in\mathbb{R}$ is a constant, 
    
    \item $a^{+}$ and $a^{-}$ are real-valued random variables on a finite probability space. 
\end{enumerate}
\end{assumption}

Assumption \ref{as:discreteassum}a) states that the request-free state evolution can be well-captured by a deterministic function that only depends on past and current weather variables. $n$ denotes the number of measured variables, and $\eta$ denotes the number of considered time steps. Despite being strong, this modeling assumption for thermal systems (in particular building assets) often leads to good results in practice because errors do not accumulate. Note that this assumption could be replaced by modeling the nominal state with a GP instead to take uncertainty into account, at the price of complicating further the analysis. Assumption \ref{as:discreteassum}b) is justified by the fact that the coefficient $b_{f}$ has little influence on the flexibility quantification discussed here, see Sections~\ref{sec:paramid} and \ref{sec:appl}. For a reasonable choice of $b_f$ see Section~\ref{sec:paramid}. Finally, Assumption \ref{as:discreteassum}c) is useful to extract the distributions of $a^{+}$ and $a^{-}$  directly from data. The random variable assumption also captures possibly random state behavior and will be helpful in the uncertainty quantification explained in Section~\ref{sec:setfeasreq}. We finally end up with the following state equation:

\begin{definition}[Battery Model]
\label{def:batmod}
Let $r_t\in \mathbb{R}$ denote a relative consumption request at time $t$ with respect to a baseline, and let $r_t^+=\max(r_t, 0), r_t^-=\min(r_t, 0)$ be the positive and negative part of the request. With the external influences given by $\bm{e}_t\in \mathbb{R}^m$ and a function $f:\mathbb{R}^m \to \mathbb{R}$ to approximate the nominal state $s_{\text{n}}$ in request-free operation, we model the state evolution as
\begin{align}
    \hat{s}_{t+1} ={ } &\hat{s}_t + a^+ r_t^+ + a^- r_t^- + b_f (f(\bm{e}_t) - \hat{s}_t) \chi_{r_t} \nonumber \\
    &+ f(\bm{e}_{t+1})-f(\bm{e}_t).
\label{eq:battery}
\end{align}
The state change depends on a parameter $b_f\in\mathbb{R}$, while $a^+$ and $a^-$ are assumed to be real-valued random variables on a finite probability space.
\end{definition}

Note that the approximated state $\hat{s}_t$ given by the battery model is no longer bounded between 0 and 1. Furthermore, $\hat{s}_t$ taking a value smaller than 0 or larger than 1 corresponds to a situation where the true state reaches its boundaries and the building controller is not able to fulfill the request.

\subsection{Parameter and Sample Space Identification}
\label{sec:paramid}

The learning of the battery model is a two-step approach. First, $f(\bm{e}_t)$ is learned from data obtained during the nominal operation of the building's controller. Then the parameter $b_f$ and the sample spaces of $a^+ $ and $ a^-$ can be identified from request periods, followed by recovery periods. 

For the learning approach, we use the following formulation that describes the dependence of the predicted state $\hat{s}_k$ on the starting state $s_0$ and the applied requests $\bm{r}_{0:k-1}$.

\begin{lemma}
\label{lem:sk}
For a given state $s_0$ and a request trajectory $\bm{r}\in\mathbb{R}^{k}$, based on \eqref{eq:battery}, the state $\hat{s}_k$ is given by
\begin{align}
   \hat s_k ={ } &(1-b_f)^{q_0^k} s_0  + \sum_{l=0}^{k-1}(1-b_f)^{q_{l+1}^k}(f(\bm{e}_l)b_f\chi_{r_l} \nonumber\\
   &+a^+ r_l^+ +a^- r_l^- + f(\bm{e}_{l+1})-f(\bm{e}_l)),
\label{eq:sk}
\end{align}
with $q_l^k=\sum_{i=l}^{k-1}\chi_{r_i}$.
\end{lemma}
\begin{proof}
Omitted for brevity. Follows from induction.
\end{proof}

For identifying the sample spaces of $a^+$ and $a^-$, we consider request sequences $\bm{r}$, either strictly positive or strictly negative respectively (i.e. $r_i>0$ or $r_i<0, i=0,\dots,k-1$). $\bm{r}$ is assumed to be followed by a request-free period, and we denote the corresponding state trajectory by $\bm{s}_{0:k}$. We either have that the requests are fulfillable, i.e. $0<s_i<1, i=0,\dots,k$ which we denote by setting an index $l=k+1$, or not fulfillable at a certain point $l$, with $l=\arg\min q \text{ s.t. } s_q=0 \text{ or } s_q=1$. Assuming a state evolution as given by \eqref{eq:sk}, we have
\begin{equation}
    s_{l-1} = s_0 + \sum_{i=0}^{l-2} a^{+/-} r_i + f(\bm{e}_{l-1})-f(\bm{e}_0).
\end{equation}
Therefore, a sample takes the form  $a^{+/-}=(s_{l-1}-f(\bm{e}_{l-1})-(s_0-f(\bm{e}_0))) /\sum_{i=0}^{l-2} r_i$. %$p=\frac{(s_{k}-f(e_k))-(s_0-f(e_0))}{\sum_{i=0}^{k-1} \bm{r}_i}$ or

To identify candidates for $b_f$, we consider sequences $\bm{s}_{0:k}$ that occur after a request period, so that $\bm{r}_{0:k-1} = \mathbf{0}$ and $r_{k}\neq 0$. Furthermore, we only use data from the recovery periods that fulfill $\vert s_t - f(\bm{e}_t)\vert > \delta$ for some threshold $\delta\in\mathbb{R}_+$, for identifying $b_f$, to capture the controller based recovery period and not small perturbations due to model mismatch. 

As in the previous cases, we either have that $\vert s_i-f(\bm{e}_i)\vert >\delta, i=0,\dots,k$ (thus $l=k+1$) or determine $l$ as $l=\arg\min q \text{ s.t. } \vert s_q-f(\bm{e}_q)\vert \leq\delta$. Using the evolution of the battery model from \eqref{eq:sk} for request-free periods, we can formulate the following least-squares problem, whose solutions give samples of the $b_f$ parameter.
\begin{equation}
    \arg\min_{b} \left( (1-b)^{l-1} (s_0 -f(\bm{e}_0))+ f(\bm{e}_{l-1}) -s_{l-1}\right)^2
\end{equation}

In the following, we denote the finite sample spaces of $a^+$ and $a^-$ as 
$\mathcal{P}^+, \mathcal{P}^-$ with $|\mathcal{P}^+| = n_1, |\mathcal{P}^-| = n_2 $. 
Furthermore, we assume an ordering, such that $\mathcal{P}^+=\{a^+_{1}, \dots, a^+_{n_1}: a^+_{i}\leq a^+_{j} \text{ if }i<j\}, \mathcal{P}^-=\{a^-_{1}, \dots, a^-_{n_2}: a^-_{i}\leq a^-_{j} \text{ if }i<j\}$, which will be helpful in Section~\ref{sec:feasreq}. Since these data are the only information we have about $a^+$ and $a^-$, it is natural to use them for constructing the sample spaces and therefore having finite sample spaces.

The parameter $b_f$ is assumed to be fixed after identifying possible candidates, e.g. by taking the maximum or average over the collected samples. This choice is deliberate because treating $b_f$ as stochastic and following through with the approach outlined in Section~\ref{sec:robustfeas} introduces combinatorial issues and nonlinearities in the uncertainty set computation while having a minimal impact on the flexibility envelope computation, due to its influence in the recovery periods only.

\section{The Set of Feasible Requests}
\label{sec:setfeasreq}

%\subsection{The set of feasible request trajectories}
In flexibility scenarios, \eqref{eq:battery} is used to determine the feasibility of request trajectories for building assets. As already stated, $\hat{s}_t$ taking a value smaller than 0 or larger than 1 corresponds to a situation where the true state saturates at its boundaries and the building controller is not able to fulfill the relative consumption request. This gives rise to the definition of the set of feasible request trajectories of length $k$, with a given probability level $\alpha$, starting from a state $s_0$:
\begin{align}
    \mathcal{R}^{\alpha}_k(s_0) = \big\{\bm{r}\in \mathbb{R}^k: &\;\mathbb{P}\left\{\mathbf{0} \leq [\hat{s}_0,  \cdots, \hat{s}_k ]^\top \leq \mathbf{1}\right\} \geq 1-\alpha, \nonumber\\ &\hat{s}_0=s_0 \big\},
    \label{eq:feasset}
\end{align}
where the probability is taken element-wise. We are especially interested in  feasible trajectories with constant relative power requests that are used in the following definition of flexibility envelopes:
\begin{definition}[Flexibility Envelope]
\label{def:flexenv}
Let $\bm p =[p_1,...,p_{n_p}]^\top \in\mathbb{R}^{n_p}$ denote a vector of discretized relative power requests and $\bm t = [t_1,...,t_{n_t}]^\top \in\mathbb{R}^{n_t}$ a vector of discrete time steps. Then we define the flexibility envelope $\bm E\in\mathbb{R}^{n_p\times n_t}$ via
\begin{align}
    \bm E_{i,j} = \max_k& \quad k \\
    \text{s.t. } &\quad \bm r_{0:k-1} \in \mathcal{R}_k^{\alpha}(f(\bm e_{t_j}))\nonumber \\
    & \quad r_l = p_i \quad \forall l=0,\dots,k-1. \nonumber
\end{align}
\end{definition}

Each entry in the flexibility envelope gives the number of timesteps that a certain relative power request $p_i$ can be sustained without violating constraints, starting from a given time $t_j$. $n_p$ is the number of desired quantization levels of the relative power requests in the flexibility envelopes, whereas $n_t$ depends on the considered horizon and time discretization. This definition gives one type of flexibility envelope, but equivalently, other types of flexibility envelopes that consider availability times of absolute consumption requests or minimum and maximum available power can be derived from \eqref{eq:feasset}. Simulation results of the flexibility envelopes defined above are presented in Section \ref{sec:simres}.

\subsection{Set Reformulation}
The goal of the following two sections is to reformulate the probabilistic set of feasible request trajectories given in \eqref{eq:feasset} into a deterministic version, using a robust uncertainty set. For this, we can rewrite \eqref{eq:sk} in the two following ways:
\begin{align}
    \hat{s}_k &= c_k + [\bm{a}^+_k, \bm{a}^-_k] \begin{pmatrix}
        \bm{r}^+_k \\ \bm{r}^-_k
    \end{pmatrix} \label{eq:skshort}\\
    &= c_k + \bm{R}_k \begin{pmatrix} a^+ \\ a^- \end{pmatrix}, \label{eq:skshort2}
\end{align}
with $\bm{r}^{+/-}_k= [r_0^{+/-}, \cdots, r_{k-1}^{+/-}]^\top \in \mathbb{R}^{k}$, where 
\begin{align}
    c_k={ }& (1-b_f)^{q_0^k} s_0  + \sum_{l=0}^{k-1}(1-b_f)^{q_{l+1}^k}(f(\bm{e}_l)b_f\chi_{r_l} \nonumber\\
    &+ f(\bm{e}_{l+1})-f(\bm{e}_l))
\end{align}
groups all the non-request parts and 
\begin{align}
 \label{eq:al}
    &\bm{a}^{+/-}_k=[(1-b_f)^{q_{1}^k}, \dots, (1-b_f)^{q_{k}^k}]  a^{+/-} \in \mathbb{R}^{k}\\
    &\bm{R}_k =\left[ \sum_{l=0}^{k-1} (1-b_f)^{q_{l+1}^k}r_l^+, \sum_{l=0}^{k-1} (1-b_f)^{q_{l+1}^k}r_l^-\right] \in \mathbb{R}^2
\end{align} 
group the request parts either depending on $a^+, a^-$, or the request trajectory, and we recall that $q_l^k=\sum_{i=l}^{k-1}\chi_{r_i}$. We can then alternatively write the set of feasible requests as
\begin{equation}
    \mathcal{R}^{\alpha}_k(s_0)
    = \bigcap_{l=0}^k \left\{\bm{r}\in\mathbb{R}^k : \mathbb{P}\left\{\bm{b}_l \leq \bm{A}_l \begin{pmatrix}
        \bm{r}^+_l \\ \bm{r}^-_l
    \end{pmatrix}\right\} \geq 1-\alpha\right\}, \label{eq:feasop2}
\end{equation}
with $\bm{b}_l=[-c_l, c_l-1]^\top\in\mathbb{R}^2$ and $\bm{A}_l = \begin{pmatrix}
        \bm{a}^+_l & \bm{a}^-_l \\ -\bm{a}^+_l & -\bm{a}^-_l
    \end{pmatrix}\in\mathbb{R}^{2\times 2l}$, by using \eqref{eq:skshort}.

In the following, we will consider a single set from the intersection in \eqref{eq:feasop2} and denote $[\bm{a}^+_l, \bm{a}^-_l]$ as $\bm{a}_l$. From Assumption \ref{as:discreteassum}c), we have that  the unknown $\bm{a}_l$ is a $\mathbb{R}^{2l}$-valued random variable on a finite probability space $(\Omega, \mathcal{F}, \mathbb{P})$ with $|\Omega|=N, \mathcal{F}=2^{\Omega}$. We can construct the support by combining all possible $a^+, a^-$ from the identified sets $\mathcal{P}^{+}, \mathcal{P}^{-}$. The sample set for $\bm{a}_l$ is denoted by $\mathcal{A}_l=\{\bm{a}_{l,1}, \dots, \bm{a}_{l,N}\}$ with $|\mathcal{A}_l|  = n_1 n_2 =:N$. The data matrix is denoted as $\bm{D}_l=[\bm{a}_{l,1}^\top,\dots,\bm{a}_{l,N}^\top]\in \mathbb{R}^{2l\times N}$.

On the one hand, having $\bm{a}_l$ as a random variable on a finite probability space is restrictive, since the true sample space $\Omega$ might be larger or even continuous. On the other hand, since data is the only knowledge we have about $\bm{a}_l$, this assumption is aligned with the data-driven approach, and useful in practice (see \cite[Assumption~3.1]{bertsimas_constructing_2009}).

\subsection{Robustness via Conditional Value at Risk}
\label{sec:robustfeas}

Utilizing the new formulation of the set of feasible trajectories \eqref{eq:feasop2}, we will now exploit a specific risk measure, the Conditional Value at Risk (CVaR), as a way to specify how the uncertainty is dealt with. Concretely, we will consider user preferences to trade off conservativeness and the size of the feasible set, relying on results from \cite{bertsimas_constructing_2009}. The presented derivations are not unique for CVaR but also hold for general coherent risk measures. However, using CVaR, together with two straightforward assumptions, gives us a directly usable, tightened version of the set of feasible requests, which is why we focus our discussions on this specific risk measure.

Here, we will only present the main concepts necessary for our specific approach. For some additional insight, the reader is referred to e.g. \cite{delbaen_2002}.

\begin{definition}[Conditional Value at Risk]
\label{def:cvar}
Let $(\Omega, \mathcal{F}, \mathbb{P})$ be a finite probability space with $\Omega=\{\omega_1, \dots,\omega_N\}$  and let $\mathcal{X}$ be a linear space of random variables on $\Omega$. We define the CVaR for $X\in\mathcal{X}$ with probability level $\alpha$ as
\begin{equation}
    \cvar_\alpha(X)=\max_{\bm{q}\in\mathcal{Q}}\frac{1}{N}\sum_{i=1}^N - q_iX(\omega_i), 
\end{equation}
with $\mathcal{Q}$ its family of generating measures, given by $\{\bm{q}\in\Delta^N: q_i\leq \frac{\mathbb{P}(\omega_i)}{\alpha}\}$.
\end{definition}

An intuition about the meaning of CVaR can be drawn from its continuous probability space definition for atomless distributions (this intuition is inexact in the finite case, but nevertheless helpful). If we consider a constraint $\bm{a}^\top \bm{x} \geq b$ for a random variable $\bm{a}\in\mathbb{R}^l$, then $\cvar_\alpha(\bm{a}^\top \bm{x} -b)$ gives the expected constraint violation in the $\alpha$-\% worst cases. This motivates the use of the risk-aversion constraint $\cvar_\alpha(\bm{a}^\top \bm{x} -b)\leq 0$. Note that this constraint implies both constraint satisfaction in expectation and constraint satisfaction with probability $\geq 1-\alpha$.

We will apply risk aversion constraints to the individual probabilistic constraints in \eqref{eq:feasop2} and utilize the reformulation with robust uncertainty sets presented in \cite[Thm. 3.1]{bertsimas_constructing_2009}. This is possible since CVaR is a coherent risk measure (i.e. it fulfills the properties of monotonicity, translation invariance, convexity, and positive homogeneity). Using risk aversion constraints instead of probabilistic constraints leads to a smaller feasible set for the same uncertainty level $\alpha$ since the former represent a tightened version of the latter.
 
\begin{theorem}
\label{thm:uncequiv}
We have
\begin{align}
    &\left\{\bm{r}\in\mathbb{R}^k: \cvar_{\alpha}\left(\bm{A}_l\begin{pmatrix}
        \bm{r}^+_l \\ \bm{r}^-_l
    \end{pmatrix}-\bm{b}_l\right)\leq 0\right\} \\
    &= \left\{\bm{r}\in\mathbb{R}^k: [\bm{a}, -\bm{a}]^\top\begin{pmatrix}
        \bm{r}^+_l \\ \bm{r}^-_l
    \end{pmatrix}\geq \bm{b}_l \quad \forall \bm{a} \in \mathcal{U}^l_{\alpha}\right\},
    \label{eq:uncertequiv}
\end{align}
with a slight abuse of notation for the constraint-wise CVaR application, where $\mathcal{U}^l_{\alpha} = \conv(\{\bm{D}_l\bm{q}:\bm{q}\in\mathcal{Q}\})$, and we recall that $\bm{D}_l$ is the data matrix, and $\mathcal{Q}$ the family of generating measures for $\cvar_{\alpha}$.
\end{theorem}
\begin{proof}
We can write the left-hand side as 
\begin{align}
    &\left\{\bm{r}\in\mathbb{R}^k: \cvar_{\alpha}\left(\bm{a}_l\begin{pmatrix}
        \bm{r}^+_l \\ \bm{r}^-_l
    \end{pmatrix}+c_l\right)\leq 0\right\} \\
    &\cap \left\{\bm{r}\in\mathbb{R}^k: \cvar_{\alpha}\left(-\bm{a}_l\begin{pmatrix}
        \bm{r}^+_l \\ \bm{r}^-_l
    \end{pmatrix}-c_l+1\right)\leq 0\right\}.
\end{align}
Using the robust uncertainty set reformulation from \cite{bertsimas_constructing_2009} Theorem 3.1. for both sets in the intersection, we get
\begin{align}
    &\left\{\bm{r}\in\mathbb{R}^k: \bm{a}^\top\begin{pmatrix}
        \bm{r}^+_l \\ \bm{r}^-_l
    \end{pmatrix}\geq -c_l \quad \forall \bm{a} \in \mathcal{U}^l_{\alpha}\right\} \\
    & \cap \left\{\bm{r}\in\mathbb{R}^k: -\bm{a}^\top\begin{pmatrix}
        \bm{r}^+_l \\ \bm{r}^-_l
    \end{pmatrix}\geq c_l-1 \quad \forall \bm{a} \in \mathcal{U}^l_{\alpha}\right\}.
\end{align}
Since the same uncertainty sets are used, we can combine them in the form of \eqref{eq:uncertequiv}.
\end{proof}

Theorem~\ref{thm:uncequiv} is not limited to CVaR, but holds for general coherent risk measures. It provides a closed form description of the set of request trajectories that fulfill the risk aversion constraint, by taking those that are robustly feasible for the uncertainty set $\mathcal{U}^l_{\alpha}$. For the feasibility of a given $\bm{r}\in\mathbb{R}^k$, this implies checking constraint satisfaction for all $\bm{a}\in\mathcal{U}^l_{\alpha}$ and $l=1,\dots, k+1$.

From Definition~\ref{def:cvar}, we can directly observe the uncertainty set construction as in Theorem~\ref{thm:uncequiv} for $\cvar_\alpha$, namely $\mathcal{U}^l_{\alpha}=\conv(\{\bm{D}_l\bm{q}: \bm{q}\in\Delta^N, q_i\leq\frac{\mathbb{P}(\bm a_{l,i})}{\alpha}\})$, with the data matrix $\bm{D}_l$. The following theorem states a more practical form of this uncertainty set under certain assumptions.
\begin{theorem}
Let the probabilities on the finite probability space be uniform (i.e. $\mathbb{P}(\bm a_{l,i})=\frac{1}{N}$), and $\alpha$ chosen as $\frac{j}{N}$ for some $j\in\{1,\dots,N\}$. Then $\mathcal{U}^l_{\alpha}$ is the convex hull of all $j$-point averages in $\mathcal{A}_l$, i.e.
\begin{equation}
    \mathcal{U}^l_{\alpha} = \conv\left(\left\{\frac{1}{j}\sum_{i\in J}\bm{a}_{l,i}: J\subset \{1,\dots,N\}, |J|=j\right\}\right).\label{eq:jpointU}
\end{equation}
\end{theorem}
\begin{proof}
The first direction, i.e. ``$\supseteq$'',
follows directly from the definition of the family of generating measures.

For ``$\subseteq$'', pick an arbitrary $\bm{a}\in \mathcal{U}^l_\alpha$ and observe that we can write it as $\bm{a}=\bm{D}_l\bm{q}$ for some $\bm{q}\in \{\bm{q}\in\Delta^N: q_i\leq \frac{1}{j}\forall i\}$. Furthermore, we have $\conv(Q)=\{\bm{q}\in\Delta^N: q_i\leq \frac{1}{j}\forall i\}$ for $Q:=\{\bm{q}\in \mathbb{R}^N: q_i = \frac{1}{j}\forall i\in J, q_i = 0  \forall i \not \in J,J\subset \{1,\dots,N\}, |J|=j \}$. Since each element $\bar{\bm{a}}=\bm{D}_l\bar{\bm{q}}$ with $\bar{\bm{q}}\in\conv (Q)$ is in the right-hand-side, we conclude that $\bm{a}$ is an element of the right-hand-side. 
\end{proof}
These uncertainty sets given by the $j$-point averages are the sets we will focus on in the following, for two reasons: Firstly, we do not consider a weighting of the samples, which makes the choice of uniform probabilities natural. Secondly, for $N$ large enough, the choice of $\alpha$ as $\frac{j}{N}$ offers a fine discretization, while also providing a straightforward way of computing the uncertainty set. The advantage of computability, therefore, outweighs the limitation of choice through discretization.

We can then formulate the tightened set of feasible request trajectories, based on the CVaR uncertainty sets as 
\begin{equation}
    \mathcal{C}^{\alpha}_k(s_0) = \bigcap_{l=0}^k \left\{\bm{r}\in\mathbb{R}^k :[\bm{a}, -\bm{a}]^\top\begin{pmatrix}
        \bm{r}^+_l \\ \bm{r}^-_l
    \end{pmatrix}\geq \bm{b}_l \quad \forall \bm{a} \in \mathcal{U}^l_{\alpha} \right\}.
    \label{eq:feassettight}
\end{equation}

\section{Application}
\label{sec:appl}
We will now use the battery model with uncertainty quantification to show how the feasibility of constant requests can be tested and used for the computation of flexibility envelopes, which is shown in a simulation example.

\subsection{Feasibility of Constant Request Trajectories}
\label{sec:feasreq}
To test whether a request trajectory $\bm{r}\in \mathbb{R}^k$ with $r_i = p, i=0,\dots,k-1$, is classified as feasible, i.e. its containment in $\mathcal{C}^\alpha_k(s_0)$, we have to build the uncertainty sets $\mathcal{U}^l_{\alpha}$ for a given $\alpha=\frac{j}{N}$, and check the constraints in \eqref{eq:feassettight}. However, characterizing the convex hull of all $j$-point averages is not straightforward, and constructing and testing all combinations is combinatorially infeasible already for moderate values of $j$ and $N$. To alleviate this issue, we make the following observation. Due to the linearity of the $\bm{a}^{+/-}_l$ from \eqref{eq:skshort} in $a^{+/-}$, we can consider the $j$-point averages of the $(a^+, a^-)$ pairs in $\mathcal{P}^+\times \mathcal{P}^-$, instead of their induced $\bm{a}_l$ samples, and test feasibility for \eqref{eq:skshort2} because of its equivalence to \eqref{eq:skshort}. We denote this set of $j$-point averages by $\mathcal{P}_j = \{\frac{1}{j} \sum_{i=1}^j\bm{a}_i: \bm{a}_i\in\mathcal{P}^+\times\mathcal{P}^-, \bm{a}_l\neq\bm{a}_k\text{ for }l\neq k\}$.

However, since the relative request is constant, either only $a^+$ or $a^-$ has to be considered regarding the feasibility problem, depending on the sign of the request. Therefore, we can restrict ourselves to testing feasibility for the worst-case parameters in $\mathcal{P}_j$. Since the sets $\mathcal{P}^{+}$ and $\mathcal{P}^{-}$ are increasingly ordered, these parameters are given by $(\tilde{a}^+, \tilde{a}^-)$ with
\begin{align}
    \tilde{a}^+ &= \frac{1}{j} \left( n_2 \sum_{i=0}^{\floor{\frac{j}{n_2}}} a^+_{n_1-i} + (j\text{ mod } n_2) a^+_{n_1 - \floor{\frac{j}{n_2}}-1}\right),\\
    \tilde{a}^- &= \frac{1}{j} \left( n_1 \sum_{i=0}^{\floor{\frac{j}{n_1}}} a^-_{n_2-i} + (j\text{ mod } n_1) a^-_{n_2 - \floor{\frac{j}{n_1}}-1}\right).
\end{align}

We use this method to quantify available flexibility for each time in a day ahead prediction with the flexibility envelopes from Definition~\ref{def:flexenv}. This is done in an iterative fashion by increasing the length of the considered trajectory, to get the maximum number of steps $k$ and then converting it to time. We limit the maximum sustainable time to 24 hours since forecast errors might deter the quality of longer predictions. An example of flexibility envelopes can be found in Fig.~\ref{fig:flex_envelopes}.

\begin{figure*}[!t]
    \centering
    \includegraphics[width=0.85\textwidth]{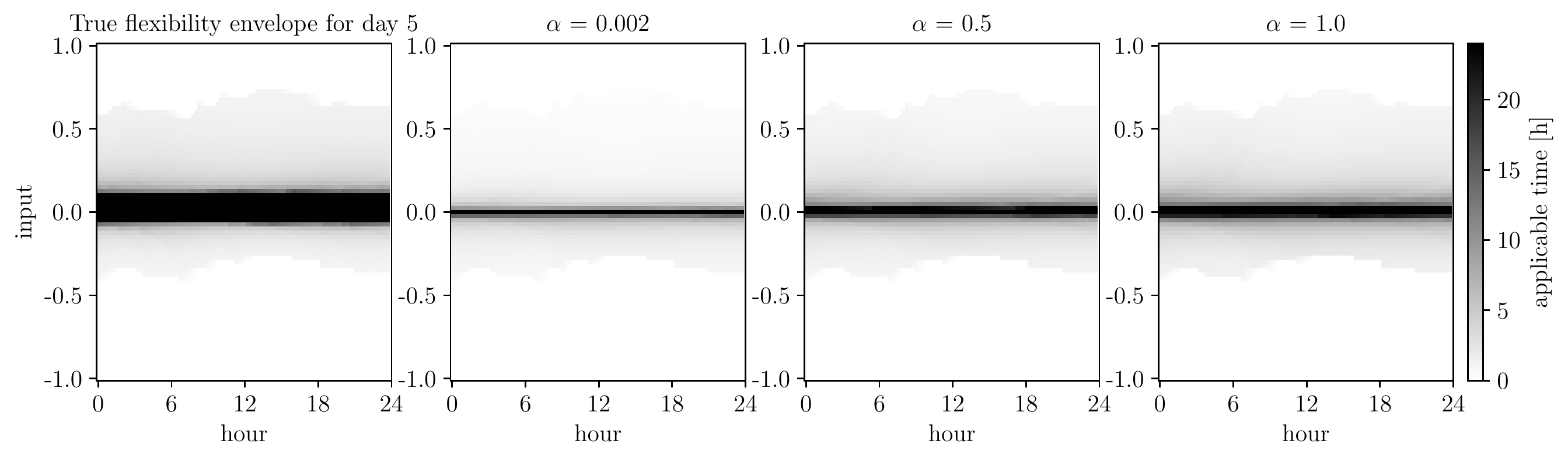}
    \caption{Flexibility envelopes for day 5 of the test data. First: True flexibility envelope. Second to fourth: Flexibility envelope predictions for different $\alpha$ values.}
    \label{fig:flex_envelopes}
\end{figure*}

\subsection{State Definition}
\label{sec:statedef}
In the experiments, we use a specific state definition that fulfills the requirements of the abstract state $s_t$ from Section \ref{subsec:model}. The two key quantities for this are the available time of running the flexible assets of the building at minimum power (denoted by $\ubar{\Delta}_t$) and the available time for maximum power (denoted by $\bar{\Delta}_t$), without violating thermal constraints. We assume that these are provided at each time step $t$ by the controller. This could either be through directly providing these quantities or by providing standard measurements from which these quantities can be derived. An example for heat pumps is given in  \cite[Eqn. 8]{gasser_predictive_2021}, with 
\begin{equation*}
    \ubar{\Delta}_t =  \frac{T_t-T_{\min}}{P_{\text{loss},t}}, \;
    \bar{\Delta}_t =  \frac{T_{\max}-T_t}{P_{\max} - P_{\text{loss},t}},
\end{equation*}
where the current temperature is given by $T_t$, the upper and lower temperature limits by $T_{\max}$ and $ T_{\min}$ respectively, the thermal power capacity of the heat pump by $P_{\max}$, and the average losses at time $t$, by $P_{\text{loss},t}$.

Following Assumption~\ref{as:statebound}, we assume that $\ubar{\Delta}_t + \bar{\Delta}_t>0$. 

\begin{definition}[State Variable]
We define the state $s_t\in \mathbb{R}$ of a building at time $t$ as 
\begin{equation}
    s_t := \frac{\ubar{\Delta}_t}{\ubar{\Delta}_t + \bar{\Delta}_t}.
    \label{eq:state}
\end{equation}
\end{definition}

By definition, we get that $s_t \in [0,1]$. Moreover, this state captures the desired properties described in Section~\ref{subsec:model}.

\subsection{Simulation Results}
\label{sec:simres}
We consider the flexibility envelope prediction for the SimpleHouseRad-v0 model from the simulation model library Energym \cite{scharnhorst_energym}. SimpleHouseRad-v0 is a lightweight Modelica model house with a 5-minute simulation timestep, modeled as a single zone, equipped with a heat pump whose electrical power fraction is the control input (i.e. it is in the range $[0,1]$). The building model is controlled by a PI controller that measures and reports the state introduced in \eqref{eq:state} based on a simplified model of the building. This PI controller is also used to track the request trajectories as close as possible without violating the temperature bounds of 19 and 24 °C. 

Data collection for constructing the battery model is performed during the first 6 weeks of a year, using measurements of the external temperature and irradiance from the city of Basel, Switzerland. In the first three weeks, no requests are sent to the building, such that state data under nominal controller operation is collected for learning the nonlinear model $f(\bm{e}_t)$, using a kernel ridge regression model with squared exponential kernel. During the second three weeks, random constant requests are sent to the building for random durations between 1 hour and 4 hours, alternating with request-free periods of 4 hours to 15 hours. Since the control input is the heat pump power fraction, we consider input requests instead of power requests. This data collection resulted in a total of 22 samples for $a^+$ and 20 samples for $a^-$, giving $N=440$ as the size of the discrete sample space.

We compute the flexibility envelopes of 10 days, starting from the 22nd of January, for a weather file from Lausanne, Switzerland. Different values of the uncertainty parameter $\alpha$ are used, and we compare the results with the true available flexibility. These true flexibility envelopes are computed by running the relative requests on the simulation model itself and observing violations of the temperature bounds.

An example of this evaluation is given in Fig.~\ref{fig:flex_envelopes} for day 5 of the chosen 10 days in the test data. A decrease in conservativeness is observable for an increase in $\alpha$.

The pointwise predicted availability (in number of timesteps) vs. the true availability of the requests, is shown in Fig.~\ref{fig:flex_accuracy}. Points that lie on the diagonal or slightly above are desirable, while points that are below the diagonal represent predictions that are more optimistic than the actual availability (and are therefore infeasible). For this specific day, we observe about $0.2\%$ infeasible predictions for $\alpha=\frac{1}{440}\approx 0.002$, while for $\alpha=0.5$ and $\alpha=1.0$ the infeasible predictions make up about $5.1\%$ and $10.7\%$ respectively.

\begin{figure*}[!t]
    \centering
    \includegraphics[width=0.85\textwidth]{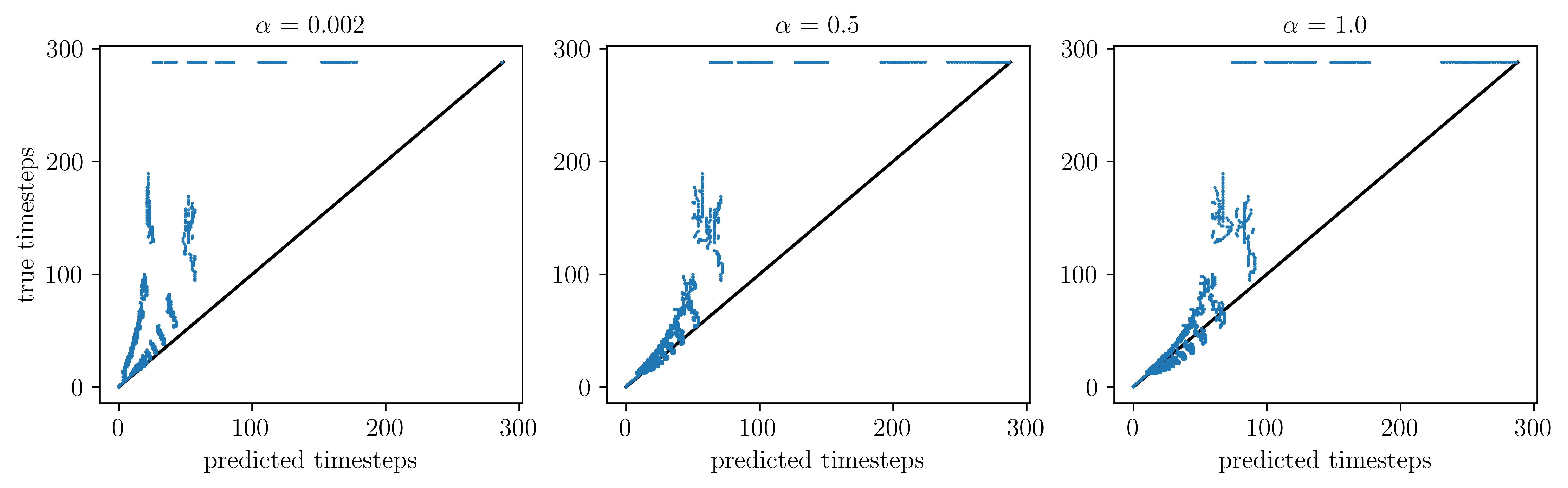}
    \caption{Predicted vs. truly available timesteps in the flexibility envelopes for different values of $\alpha$ on day 5 of the test data.}
    \label{fig:flex_accuracy}
\end{figure*}

\begin{figure}[!t]
    \centering
    \includegraphics[width=0.6\columnwidth]{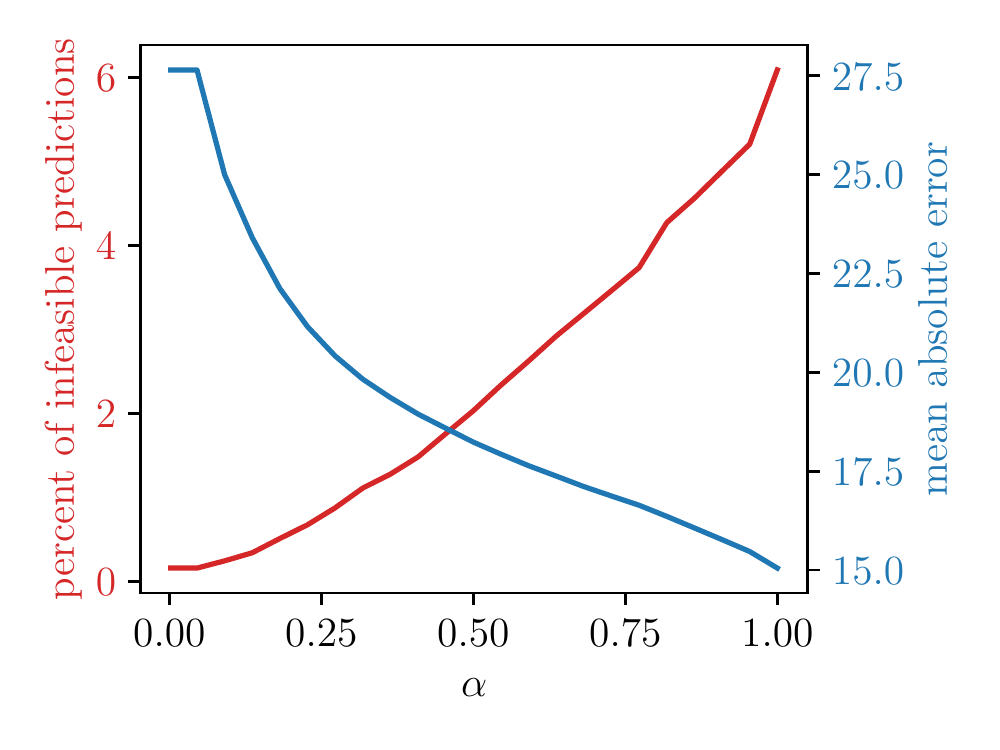}
    \caption{Percentage of infeasible prediction vs mean absolute prediction error of the flexibility envelopes for different values of $\alpha$ over 10 days.}
    \label{fig:flex_err}
\end{figure}

We get the following results regarding the percentage of infeasible predictions and mean absolute prediction error, displayed in Fig.~\ref{fig:flex_err}. The percentage of infeasible predictions, over the course of the 10 days, is at about 0.16\% for $\alpha=\frac{1}{440}$, and it rises up to about 6.09\% for $\alpha=1.0$. On the other hand, the mean absolute prediction error decreases from about 28 timesteps for $\alpha=\frac{1}{440}$, to about 15 for $\alpha=1.0$, indicating a tradeoff between conservativeness and prediction error.

This tradeoff, together with the incentives for providing flexibility and penalties for not being able to provide the promised flexibility, can inform the selection of an uncertainty parameter $\alpha$ to be used in a flexibility scheme.

\section{Conclusion}
\label{sec:con}
We presented a novel battery model formulation for estimating the available flexibility of a single building. The key features of this battery model are its physical motivation, its data-driven nature, and its uncertainty quantification. We demonstrated how to handle the uncertainty by incorporating risk preferences of the user through risk measures and how to compute the uncertainty set efficiently. The approach was tested for the flexibility envelope prediction of a simulation model from the Python library Energym, highlighting the influence of the user defined risk parameter.

In the battery model formulation, a general model for predicting the nominal state was used. An evaluation of different model structures, in dependence on the used controller, is of interest for future work. This goes hand in hand with the incorporation of uncertainty in the external influences (see e.g. \cite{MADDALENA2021109896}), which will be the topic of future work. Finally, this work focused on single buildings. For a subsequent study, we aim to aggregate multiple assets and coordinate the dispatch of incoming aggregated flexibility requests.

%\IEEEtriggeratref{12}
\bibliographystyle{ieeetr}
\bibliography{root}

\end{document}